\documentclass[ejsv2, noshowframe]{imsart}

\RequirePackage[authoryear]{natbib}
\PassOptionsToPackage{hyphens}{url}
\RequirePackage[colorlinks,citecolor=blue,urlcolor=blue]{hyperref}

\usepackage[utf8]{inputenc}
\usepackage[T1]{fontenc}
\usepackage{amsmath,amsthm}

\usepackage{amssymb}
\usepackage{algorithm}
\usepackage{algpseudocode}
\usepackage{float}
\usepackage{graphicx}
\graphicspath{{../experiment/core/out/}}
\usepackage{booktabs}
\usepackage{multirow}
\usepackage{xcolor}
\usepackage{placeins}

\startlocaldefs
\theoremstyle{plain}
\newtheorem{theorem}{Theorem}
\newtheorem{proposition}{Proposition}

\theoremstyle{definition}

\theoremstyle{remark}

\newcommand{\E}{\mathbb{E}}
\newcommand{\Prob}{\mathbb{P}}

\newcommand{\Unif}{\mathrm{Uniform}}
\DeclareMathOperator*{\argmax}{arg\,max}

\newcommand{\expEpochs}{500}

\newcommand{\expNtrain}{10{,}000}
\newcommand{\expNstable}{2{,}500}
\newcommand{\expNpost}{2{,}500}
\newcommand{\expAlpha}{0.05}
\newcommand{\expNbins}{100}
\newcommand{\expNtrials}{10{,}000}

\newcommand{\resPMGRATPR}{96.2\%}
\newcommand{\resPMGRAFPR}{3.8\%}
\newcommand{\resPMGRADelay}{77}
\newcommand{\resPMGRACPErr}{1.2}
\newcommand{\resPMGSGTPR}{96.2\%}

\newcommand{\resPMGSGDelay}{189}
\newcommand{\resPMGSGCPErr}{6.8}
\newcommand{\resPMLEATPR}{96.2\%}

\newcommand{\resPMLEADelay}{1919}
\newcommand{\resPMLEACPErr}{1839.3}

\newcommand{\resADGRATPR}{99.1\%}
\newcommand{\resADGRAFPR}{0.9\%}
\newcommand{\resADGRADelay}{27}
\newcommand{\resADGSGTPR}{99.1\%}

\newcommand{\resADGSGDelay}{27}
\newcommand{\resADLEATPR}{99.1\%}

\newcommand{\resADLEADelay}{115}

\newcommand{\resDDMGRATPR}{90.8\%}
\newcommand{\resDDMGRAFPR}{9.2\%}
\newcommand{\resDDMGRADelay}{405}
\newcommand{\resDDMGSGTPR}{90.0\%}

\newcommand{\resDDMGSGDelay}{666}

\newcommand{\resHDDMAGRATPR}{94.0\%}
\newcommand{\resHDDMAGRAFPR}{6.0\%}
\newcommand{\resHDDMAGRADelay}{60}
\newcommand{\resHDDMAGSGTPR}{94.0\%}

\newcommand{\resHDDMAGSGDelay}{168}

\newcommand{\addLocMAEOnset}{1808.2}
\newcommand{\addLocMAEMax}{265.6}

\newcommand{\addLocCloserMax}{9{,}629}
\newcommand{\addLocN}{10{,}000}

\endlocaldefs

\begin{document}
\begin{frontmatter}

\title{When Your Model Stops Working:\\ Anytime-Valid Calibration Monitoring}
\runtitle{Anytime-Valid Calibration Monitoring}

\begin{aug}
\author[A]{\fnms{Tristan}~\snm{Farran}}
\address[A]{University of Amsterdam}
\runauthor{T. Farran}
\end{aug}

\begin{abstract}
Practitioners monitoring deployed probabilistic models face a fundamental trap: any fixed-sample test applied repeatedly over an unbounded stream will eventually raise a false alarm, even when the model remains perfectly stable. Existing methods typically lack formal error guarantees, conflate alarm time with changepoint location, and monitor indirect signals that do not fully characterize calibration. We present PITMonitor, an \emph{anytime-valid} calibration-specific monitor that detects distributional shifts in probability integral transforms via a mixture e-process, providing Type~I error control over an unbounded monitoring horizon as well as Bayesian changepoint estimation. On \texttt{river}'s FriedmanDrift benchmark, PITMonitor achieves detection rates competitive with the strongest baselines across all three scenarios, although detection delay is substantially longer under local drift. Code is available at \url{https://github.com/tristan-farran/pitmon}.
\end{abstract}



\end{frontmatter}

\section{Introduction}

Probabilistic models deployed in production face a fundamental challenge: the world changes. Across domains, from medicine to finance, models encounter regime shifts and concept drift which can cause calibration to degrade drastically, with consequential effects downstream. Modern neural networks, for instance, can improve classification accuracy while simultaneously becoming less calibrated \citep{guo2017calibration}. Practitioners must therefore monitor their models continuously, seeking to answer two questions: \emph{has} the world changed, and if so, \emph{when}? In practice, this monitoring often relies on ad-hoc procedures such as periodic recalibration schedules, rolling-window hypothesis tests, or threshold-based alerts on summary metrics.

These approaches suffer from a fundamental statistical problem: \emph{they do not control the false alarm rate over continuous monitoring}. A practitioner who checks calibration daily with a $p < 0.05$ threshold will, over a year of monitoring, almost certainly observe spurious alarms even if calibration remains stable. Classical hypothesis tests assume a fixed sample size determined before seeing data; continuous monitoring violates this assumption.

More principled alternatives are provided by online drift detectors, such as those implemented in the \texttt{river} library \citep{montiel2021river}. Classical detectors such as DDM and HDDM are lightweight and often effective at detecting abrupt changes, but rely on heuristic thresholds or fixed-sample statistical arguments that do not provide false alarm guarantees under continuous monitoring. ADWIN \citep{bifet2007learning} improves on fixed-window methods by adapting its window size to bound the false alarm probability \emph{within a single window}, but this per-window guarantee does not extend to the stream level: as the stream grows, the number of implicit tests accumulates without bound, inflating the false alarm rate. More fundamentally, these methods typically monitor generic signals such as error rates or mean changes in residuals rather than calibration-specific quantities, with broader notions of miscalibration such as overconfidence with unchanged accuracy potentially going undetected.

We propose PITMonitor, an anytime-valid calibration monitor with five key properties:
\begin{enumerate}
 \item \textbf{Anytime-valid false alarm control:} we prove that $\Prob(\text{ever alarm} \mid H_0) \leq \alpha$ for all time, without requiring a pre-specified monitoring horizon or stopping rule.
 
 \item \textbf{Calibration-specificity:} we target the full calibration relationship between model and data, rather than indirect signals that may miss important aspects of miscalibration.
 
 \item \textbf{No baseline period required:} unlike methods requiring a ``clean'' reference distribution or initialisation window, PITMonitor works from the first observation.

 \item \textbf{Practical efficiency:} the method runs in $O(1)$ per update of the mixture e-process, plus $O(\log n)$ for rank maintenance via a \texttt{SortedList} \citep{sortedcontainers}, with $O(n)$ memory.

 \item \textbf{Distribution-free:} no assumptions are made about the pre- or post-change distributions.
\end{enumerate}

\section{Related Work}

\subsubsection*{Calibration Assessment}
Classical calibration metrics include expected calibration error \citep{guo2017calibration}, reliability diagrams \citep{degroot1983comparison}, and proper scoring rules \citep{gneiting2007strictly}. PITs have been used for static forecast evaluation in econometrics \citep{diebold1998evaluating} and weather forecasting \citep{gneiting2014probabilistic}. These both provide static calibration assessment but do not address sequential monitoring with false alarm control.

\subsubsection*{Distribution Shift Detection}
Methods for detecting distribution shift include two-sample tests \citep{rabanser2019failing} and domain classifiers \citep{lipton2018detecting}. These methods primarily address changes in the input or label distributions, while our work focuses on the \emph{output} side: identifying when predicted probabilities no longer match outcome frequencies.

\subsubsection*{Sequential Calibration Testing}
\citet{arnold2023sequentially} proposed e-values for testing forecast calibration, focusing on whether PITs are $\Unif(0,1)$. PITMonitor differs in two ways. First, we do not assume a specific PIT distribution and instead test whether the PIT distribution remains stable over time, alarming only on \emph{shifts} in calibration rather than on miscalibration itself. Second, we locate \emph{when} calibration changed, rather than only establishing \emph{whether} the model is miscalibrated.

\subsubsection*{E-values}
E-values support anytime-valid hypothesis testing under optional stopping \citep{vovk2021values, ramdas2023game, grunwald2024safe, shin2022detectors}. PITMonitor builds on this framework, constructing a mixture e-process to detect shifts in the PIT distribution and inheriting its anytime-valid false-alarm guarantee (Theorem~\ref{thm:type1}).

\section{Background}

\subsection{Probability Integral Transforms}
\label{pits}
A probabilistic model with predictive CDF $\hat{F}$ for $Y \sim P$ is \emph{perfectly calibrated} if $\hat{F} = P$. \emph{Miscalibration}, meaning divergence of predicted probabilities from outcome frequencies, can manifest as systematic bias, over- or under-confidence, misassigned tail risks, and more. The \emph{probability integral transform} (PIT) $U = \hat{F}(Y)$ unifies all of these into a single object computable from predictions and observations alone.

Since $U$ has distribution $P \circ \hat{F}^{-1}$, it encodes the calibration relationship between the model and data-generating process: every discrepancy between $\hat{F}$ and $P$ leaves a signature in the distribution of $U$. Perfect calibration corresponds to $U \sim \Unif(0,1)$ \citep{dawid1984statistical}, but the shape of the PIT distribution is informative more broadly: a U-shaped density indicates underdispersion, a hump near the center overdispersion, and a triangular density systematic location bias \citep{gneiting2007probabilistic}. Thus, shifts in the PIT distribution indicate changes in calibration, making it a natural monitoring target. We therefore test:
\begin{itemize}
 \item $\mathbf{H_0}$: $U_1, U_2, \ldots \overset{\mathrm{i.i.d.}}{\sim} F$ for some distribution $F$
 \item $\mathbf{H_1}$: There exists a $\tau > 1$ such that $(U_t)_{t < \tau} \overset{\mathrm{i.i.d.}}{\sim} F$ and $(U_t)_{t \ge \tau}$ is not i.i.d.\ from $F$
\end{itemize}
As a result, PITMonitor is able to target any change in the calibration relationship, without alarming on static miscalibration. The null is broad, encompassing both perfectly calibrated models ($F = \Unif(0,1)$) and miscalibrated yet stable ones ($F \neq \Unif(0,1)$), while the alternative captures any change in calibration, whether abrupt or gradual, local or global.

\subsection{Conformal P-values}
\label{sec:pvals}

To detect distribution shifts sequentially, we employ \emph{conformal p-values from ranks} \citep{vovk2005algorithmic}. Given observations $U_1, \ldots, U_t$, define the rank of $U_t$ as:
\begin{equation}
 R_t = \#\{s \leq t : U_s \leq U_t\}
\end{equation}

\begin{proposition}[Rank Uniformity Under $H_0$]\footnote{This result holds more generally under exchangeability of $(U_t)$; we assume i.i.d.\ throughout for simplicity.}
\label{prop:rank}
If $U_1, \ldots, U_t \overset{i.i.d.}{\sim} F$ with $F$ continuous,\footnote{If $F$ is discretized or rounded, randomized tie-breaking can be used to preserve uniformity.} then $R_t \sim \mathrm{Unif}\{1, \ldots, t\}$.
\end{proposition}
\begin{proof}
Since $F$ is continuous, ties occur with probability zero. By exchangeability, all $t!$ orderings of $(U_1,\ldots,U_t)$ are equally likely. Exactly $(t-1)!$ of these orderings place $U_t$ at rank $r$, giving $\Prob(R_t = r) = 1/t$.
\end{proof}

\noindent Since $R_t$ is discrete, we add an independent jitter $V_t$ to obtain a $p_t \sim \Unif(0,1)$ under $H_0$:
\begin{equation}
 p_t = \frac{R_t - 1 + V_t}{t}, \quad V_t \sim \Unif(0,1)
\end{equation}
Under $H_1$, new observations no longer fall uniformly across the reference sample, making the p-value distribution a useful signal of change. Intuitively, post-change p-values are displaced toward regions where the new law concentrates more mass than $F$.

One subtlety is that this non-uniformity is transient: as we accumulate post-change observations, $p_t$ becomes increasingly uniform even under $H_1$ and the signal erodes. A formal characterization of post-change p-value behaviour is left for future work (Section~\ref{sec:conclusion}).

\subsection{E-values}
An \emph{e-value} is a nonnegative random variable satisfying $\E[e_t \mid \mathcal{F}_{t-1}] \leq 1$ under the null hypothesis. A key property of e-values is that they can be composed multiplicatively while maintaining validity under the null \citep{ramdas2023game}. The resulting product is an \emph{e-process}, accumulating evidence from time $\tau$ onwards:
\begin{equation}
 M_t^{(\tau)} = \prod_{s=\tau}^{t} e_s
\end{equation}
Under the alternative, e-values are constructed so that $\E[e_t \mid \mathcal{F}_{t-1}] > 1$. Therefore,
\begin{equation}
 \E[M_t^{(\tau)} \mid \mathcal{F}_{t-1}] = M_{t-1}^{(\tau)} \cdot \E[e_t \mid \mathcal{F}_{t-1}]
 \begin{cases} \leq M_{t-1}^{(\tau)} & \text{under } H_0, \\ > M_{t-1}^{(\tau)} & \text{under } H_1. \end{cases}
\end{equation}
Hence, $M_t^{(\tau)}$ is a nonnegative supermartingale under $H_0$, yet grows without bound in expectation under $H_1$. This motivates detecting a shift by monitoring whether the process exceeds a threshold. By Ville's inequality \citep{ville1939}, $\Prob(\sup_{t \geq \tau} M_t^{(\tau)} \geq 1/\alpha) \leq \alpha$ under $H_0$, so the threshold $1/\alpha$ controls the false alarm probability at level $\alpha$ over an unbounded horizon (as we formalise in Theorem~\ref{thm:type1}). We therefore raise an alarm at the first time $t$ such that $M_t^{(\tau)} \geq 1/\alpha$.

\section{Method}

\subsection{E-values via Density Betting}
\label{sec:density}

To construct e-values from p-values, we use the betting framework of \citet{ramdas2023game}. At time $t$, before observing $p_t$, we choose a density function $\hat{f}_t$ satisfying $\int_0^1 \hat{f}_t(p)\,dp=1$. Interpreting this as a wager against $H_0$, we pay one unit up front and receive payoff $\hat{f}_t(p_t)$ after $p_t$ is revealed, giving the one-step e-value $e_t = \hat{f}_t(p_t)$.

\begin{proposition}[Density Betting Yields Valid E-values]
\label{prop:density}
Let $\hat{f}: [0,1] \to [0,\infty)$ be a density function satisfying $\int_0^1 \hat{f}(p)\,dp = 1$. If $p \sim \Unif(0,1)$, then $e = \hat{f}(p)$ satisfies $\E[e] = 1$.
\end{proposition}
\begin{proof}
Since $p \sim \Unif(0,1)$, $\E[e] = \E[\hat{f}(p)] = \int_0^1 \hat{f}(p)\,dp = 1$.
\end{proof}

\noindent To bet on p-values adaptively, we use a plug-in histogram density estimator which places mass proportional to where past p-values concentrated (for the broader plug-in lineage, see \citet{fedorova2012plugin}). Letting $c_b$ count past p-values in bin $b$ and $B$ be the number of bins:
\begin{equation}
 \hat{f}(p) = B \cdot \frac{c_b}{\sum_j c_j} \quad \text{for } p \in \text{bin } b
\end{equation}
The histogram is initialized with pseudocounts $c_b = 1$ for all $b$, ensuring $\hat{f}_t$ is a valid density from the first observation, and updated \emph{after} computing $e_t$, so that $\hat{f}_t$ depends only on $p_1, \ldots, p_{t-1}$ and is therefore $\mathcal{F}_{t-1}$-measurable. Since p-values are $\Unif(0,1)$ under $H_0$ by Proposition~\ref{prop:rank}, this produces a valid e-value by Proposition~\ref{prop:density}.

Under $H_1$, the distribution shifts and p-values become non-uniform (Section \ref{sec:pvals}). As these post-change p-values fall more frequently in certain bins, the histogram learns their concentration pattern and bets accordingly, yielding $\E[e_t \mid \mathcal{F}_{t-1}] > 1$, with larger shifts producing faster e-process growth (as formalised under idealised conditions in Proposition~\ref{prop:finite_time}).

\begin{proposition}[Finite-Time Mean Gain]
\label{prop:finite_time}
Let $\tau < t$ be a changepoint and define

\begin{itemize}
 \item $A_b := \#\{s < \tau : p_s \in \mathrm{bin}\,b\}$, \quad $m := \sum A_b = \tau - 1$ \hfill (pre-shift counts),
 \item $C_b := \#\{\tau \le s < t : p_s \in \mathrm{bin}\,b\}$, \quad $n := \sum C_b = t - \tau$ \hfill (post-shift counts),
 \item $\hat q_b := (1+A_b)/(B+m)$ \hfill (pre-shift density).
\end{itemize}
Assume post-shift p-values are i.i.d.\ with bin probabilities
$\theta=(\theta_1,\ldots,\theta_B)$. Define
$$
\gamma := B\sum_{b=1}^B \theta_b \hat q_b - 1,
\qquad
\delta := \sum_{b=1}^B \theta_b^2 - \frac{1}{B},
$$
where $\gamma$ measures the alignment of $\theta$ with the pre-shift prior (with $\gamma=0$ if $\hat q_b=1/B$ for all $b$, i.e.\ a uniform prior), and $\delta$ measures shift intensity (with $\delta=0$ iff $\theta \sim \Unif(0,1)$). Conditioning on pre-shift counts,

\begin{equation}
\label{eq:finite_gain}
\E[e_t\mid A] = 1 + \frac{B+m}{B+m+n}\,\gamma + \frac{n}{B+m+n}\,B\delta.
\end{equation}
\end{proposition}
\begin{proof}
The estimator and e-value are
$$
\hat\theta_{t,b} = \frac{1 + A_b + C_b}{B + m + n},
\qquad
e_t = B\hat\theta_{t,b_t}.
$$
Since $p_t$ independently falls in bin $b$ with probability $\theta_b$,
$$
\E[e_t \mid A] = \frac{B}{B+m+n} \sum_{b=1}^B \theta_b \, \E[1 + A_b + C_b \mid A].
$$
Under the i.i.d.\ post-shift model, $\E[C_b \mid A] = n\theta_b$, so
$$
\E[e_t \mid A] = \frac{B}{B+m+n} \Bigl[\underbrace{\sum_{b=1}^B \theta_b(1+A_b)}_{(B+m)(\gamma+1)/B} + n \underbrace{\sum_{b=1}^B \theta_b^2}_{\delta + 1/B}\Bigr].
$$
Expanding and collecting gives \eqref{eq:finite_gain}.
\end{proof}

\noindent Proposition~\ref{prop:finite_time} shows the mean gain interpolating between two regimes: for small $n$ the prior-alignment term $\gamma$ dominates because the histogram has not yet learned the post-change distribution and bets according to the pre-shift pattern, while as $n\to\infty$ the data weight approaches one and $\E[e_t\mid A]\to 1+B\delta$, with the strength of the shift determining the growth of the e-process. The transition speed is governed by $n$ relative to $B+m$, with longer pre-shift histories delaying convergence, as the histogram is more entrenched in the pre-shift pattern.

When $\gamma \ge 0$, the prior is aligned and $\E[e_t\mid A]>1$ immediately; if $\gamma<0$, the expected e-value initially falls below 1. In this case, crossing $1$ requires
\begin{equation}
 n > n^* := \frac{(B+m)|\gamma|}{B\delta}
\end{equation}
This warmup length grows with the prior misalignment $|\gamma|$ and decreases with shift intensity $\delta$ (a larger shift overwhelms a bad prior faster). Since $|\gamma|\le m/(B+m)$, the worst-case warmup is bounded by $m/(B\delta)$, meaning longer null histories require more post-change observations before the e-process begins to grow, but this bound shrinks as the shift intensity $\delta$ increases.

\noindent One important caveat is that Proposition~\ref{prop:finite_time} assumes post-change p-values are i.i.d.\ with fixed bin probabilities, but as discussed in Section~\ref{sec:pvals}, p-values drift back toward uniformity over time, and since the histogram is updated with these same p-values, it learns a progressively weaker signal. Detection must therefore occur within a finite window before the signal erodes. Nonetheless, the qualitative conclusions are useful: under $H_1$, the e-process grows after a warmup period whose length decreases with shift intensity and increases with null history length, with prior alignment yielding immediate growth before the histogram has adapted.

\subsection{The Mixture E-process}
\label{sec:method_mixture}

A key challenge in constructing our e-process is that the true changepoint time is unknown. An e-process starting at $\tau$ would be sensitive to drift beginning at $\tau$ but would miss earlier changes, while one starting too early would accumulate excess noise that dilutes its power. Since a weighted mixture of e-processes is itself a valid e-process by linearity of expectation, we address the unknown changepoint time by maintaining one over all possible start times:
\begin{equation}
 M_t = \sum_{\tau=1}^{t} w_\tau \cdot M_t^{(\tau)}
\end{equation}
We use $w_\tau = 1/(\tau(\tau+1))$, as this sequence satisfies $\sum_{\tau=1}^{\infty} w_\tau = 1$ exactly, enabling an efficient $O(1)$ recursion that avoids maintaining a separate e-process for each $\tau$.
\begin{proposition}[Efficient Recursion]
The mixture e-process satisfies:
\begin{equation}
 M_t = e_t \cdot (M_{t-1} + w_t)
\end{equation}
\end{proposition}
\begin{proof}
The recursion is initialized with $M_0 = 0$. For $t \ge 1$, expand the definition:
\begin{align*}
M_t &= \sum_{\tau=1}^{t} w_\tau \cdot M_t^{(\tau)} \\
&= \sum_{\tau=1}^{t-1} w_\tau \cdot e_t \cdot M_{t-1}^{(\tau)} + w_t \cdot e_t \\
&= e_t \left( \sum_{\tau=1}^{t-1} w_\tau \cdot M_{t-1}^{(\tau)} + w_t \right) \\
&= e_t (M_{t-1} + w_t)
\end{align*}
\end{proof}

\noindent The weights also govern detection sensitivity; since $M_t \ge w_\tau M_t^{(\tau)}$ for any $\tau \le t$, the mixture fires an alarm whenever any component $M_t^{(\tau)}$ reaches the scaled threshold $\tau(\tau+1)/\alpha$, with processes starting at later times requiring proportionally stronger evidence to trigger an alarm.

\subsection{Type~I Error Control}

\begin{theorem}[Anytime-Valid False Alarm Control]
\label{thm:type1}
Under $H_0$, PITMonitor satisfies:
\begin{equation}
 \Prob\left(\sup_{t \geq 1} M_t \geq \frac{1}{\alpha}\right) \leq \alpha
\end{equation}
\end{theorem}
\begin{proof}
The mixture $M_t = \sum_{\tau=1}^{t} w_\tau M_t^{(\tau)}$ is defined with $M_t^{(\tau)}$ only for $\tau \leq t$.

\noindent To apply Ville's inequality, we extend each process to all $t \geq 0$ by defining, for each $\tau \geq 1$:
\begin{equation*}
\widetilde{M}_t^{(\tau)} =
\begin{cases}
1, & t < \tau,\\[2mm]
\prod_{s=\tau}^{t} e_s, & t \ge \tau
\end{cases}
\end{equation*}
Each $(\widetilde{M}_t^{(\tau)})_{t \ge 0}$ is a nonnegative supermartingale with $\widetilde{M}_0^{(\tau)} = 1$ (since $0 < \tau$ for all $\tau \ge 1$). Define the full mixture over all $\tau \ge 1$:
$$
\widetilde{M}_t = \sum_{\tau=1}^{\infty} w_\tau\,\widetilde{M}_t^{(\tau)}
$$
At $t=0$, since $\sum_{\tau=1}^{\infty} w_\tau = 1$, we have $\widetilde{M}_0 = 1$. Since each $w_\tau \widetilde{M}_t^{(\tau)}$ is nonnegative, we get
\begin{align*}
\E[\widetilde{M}_t \mid \mathcal{F}_{t-1}]
&= \sum_{\tau=1}^{\infty} w_\tau\, \E[\widetilde{M}_t^{(\tau)} \mid \mathcal{F}_{t-1}]
\le \sum_{\tau=1}^{\infty} w_\tau\, \widetilde{M}_{t-1}^{(\tau)}
= \widetilde{M}_{t-1}
\end{align*}
so $(\widetilde{M}_t)_{t \ge 0}$ is a nonnegative supermartingale with $\widetilde{M}_0 = 1$. Since $\widetilde{M}_t^{(\tau)} = 1$ for $\tau > t$,
\begin{align*}
\widetilde{M}_t
&= \sum_{\tau=1}^{t} w_\tau \prod_{s=\tau}^{t} e_s
 + \sum_{\tau=t+1}^{\infty} w_\tau
= M_t + \sum_{\tau=t+1}^{\infty}\frac{1}{\tau(\tau+1)} = M_t + \frac{1}{t+1}
\end{align*}
Since $\frac{1}{t+1} > 0$ for all $t \geq 1$, we have $M_t \le \widetilde{M}_t$, and therefore 
$$
\sup_{t \ge 1} M_t \le \sup_{t \ge 1} \widetilde{M}_t
$$
Applying Ville's inequality,
\begin{equation*}
\Prob\!\left(\sup_{t\ge1} M_t \ge \frac{1}{\alpha}\right)
\le
\Prob\!\left(\sup_{t\ge1} \widetilde{M}_t \ge \frac{1}{\alpha}\right)
\le \alpha. \qedhere
\end{equation*}
\end{proof}

\subsection{Changepoint Estimation}
\label{sec:changepoint}

After an alarm at time $T$, we estimate the changepoint location by selecting the split that best explains the post-split p-values as non-uniform. For each candidate $s \in \{1, \ldots, T-1\}$, we evaluate the segment $(p_{s+1}, \ldots, p_T)$ under two hypotheses, analogous to those in Section~\ref{pits}:
\begin{itemize}
 \item $\mathbf{H_0^{(s)}}$: $p_t \sim \mathrm{Unif}(0,1)$, 
 i.e.\ bin probabilities $\theta = (1/B, \ldots, 1/B)$
 \item $\mathbf{H_1^{(s)}}$: bin probabilities are unknown, with a $\mathrm{Dir}(\kappa, \ldots, \kappa)$ prior
\end{itemize}
A segment beginning at a true changepoint should be composed of p-values that are largely non-uniform, so $H_1$ should explain them better than $H_0$ as quantified by the log Bayes factor:
\begin{equation}
 \log\mathrm{BF}(s) = \log p(\mathbf{c}^{(s)} \mid H_1^{(s)}) - \log p(\mathbf{c}^{(s)} \mid H_0^{(s)})
\end{equation}
where $\mathbf{c}^{(s)} = (c_1^{(s)}, \ldots, c_B^{(s)})$ are the histogram counts for the segment under consideration. We set $\kappa = 1/2$, giving the canonical uninformative prior for the multinomial \citep{jeffreys1961}. The estimated changepoint is then identified simply as the split with the strongest evidence:
\begin{equation}
 \hat{\tau} = \argmax_s \log\mathrm{BF}(s)
\end{equation}

\section{Experiments}
\label{sec:experiments}

We evaluate PITMonitor on FriedmanDrift, a synthetic regression stream designed for controlled evaluation of drift detection methods, comparing against all seven \texttt{river} detectors. We report TPR, FPR, and detection delay for all methods, as well as changepoint estimation error for PITMonitor only, across three qualitatively distinct shift scenarios and \expNtrials{} trials.

\subsection{Setup}

\paragraph{Scenarios.}
Each scenario shares a common pre-drift distribution but differs post-onset:
\begin{itemize}
 \item \textbf{GRA} (Global Recurring Abrupt): An abrupt shift affects all features at the onset, representing a sudden global regime change. By default, GRA includes a second changepoint after which the pre-drift distribution recurs, but since all detectors halt after their first alarm it is never reached in practice; we therefore place it beyond the simulation window.
 \item \textbf{GSG} (Global Slow Gradual): A gradual shift affects all features with a 500-sample transition window during which samples are drawn from both the pre- and post-drift distributions with equal probability. As with GRA, the second changepoint is removed.
 \item \textbf{LEA} (Local Expanding Abrupt): Three abrupt changepoints are spaced evenly across the post-drift window, each expanding the affected regions of the input space. Since PITMonitor is designed around a single-changepoint assumption, LEA serves as a stress test to determine whether the method remains useful when that assumption is violated.
\end{itemize}

\paragraph{Stream layout.}
We generate a single training stream of \expNtrain{} pre-drift samples as well as \expNtrials{} independent test streams, each containing \expNstable{} pre-drift and \expNpost{} post-drift samples.

\paragraph{Predictive model.}
We train a feedforward neural network outputting a Gaussian predictive distribution $\mathcal{N}(\mu_t, \sigma_t^2)$ for each input. The network has 3 hidden layers of 128 units with SiLU activations. Inputs and targets are standardized using per-feature means and standard deviations fitted on the training set. Training uses mini-batches of size 256, the Adam optimizer with initial learning rate $3 \times 10^{-4}$, a cosine annealing learning rate schedule, and \expEpochs{} epochs. The model achieves $R^2 = 0.96$ on a held-out pre-drift test set, with an expected calibration error (ECE) of 0.01, confirming it is well-specified and calibrated before monitoring begins.

\paragraph{PIT construction.}
For each monitoring sample $(x_t, y_t)$ the PIT is:
\begin{equation}
 U_t = \Phi\!\left(\frac{y_t - \mu_t}{\sigma_t}\right)
\end{equation}
with $\mu_t$, $\sigma_t$ being the predicted mean and standard deviation and $\Phi$ the standard normal CDF.

\paragraph{Detector settings.}
All \texttt{river} baselines are run with library-default parameters, reflecting an out-of-the-box comparison that avoids detector-specific tuning. PITMonitor is run with significance level $\alpha = \expAlpha$ and $B = \expNbins$ bins. The four binary-input baselines (DDM, EDDM, HDDM\_A, HDDM\_W) require a binary error signal; we binarize via $b_t = \mathbf{1}[|r_t| > \phi]$ with $\phi$ being the median of $|r_t|$ on the training stream, the most assumption-free choice given that these methods were designed for classification and no domain guidance exists for thresholding.

\subsection{Results}
\begin{table}[t]
\centering
\small
\caption{Drift detection results on FriedmanDrift.}
\label{tab:results}
\begin{tabular}{l ccc ccc ccc}
\toprule
\multicolumn{1}{l}{} & \multicolumn{3}{c}{GRA} & \multicolumn{3}{c}{GSG} & \multicolumn{3}{c}{LEA} \\
\cmidrule(lr){2-4} \cmidrule(lr){5-7} \cmidrule(lr){8-10}
Method & TPR & FPR & Delay & TPR & FPR & Delay & TPR & FPR & Delay \\
\midrule
PITMonitor & 96.2\% & 3.8\% & 77 & 96.2\% & 3.8\% & 189 & 96.2\% & 3.8\% & 1919 \\
ADWIN & 99.1\% & 0.9\% & 27 & 99.1\% & 0.9\% & 27 & 99.1\% & 0.9\% & 115 \\
KSWIN & 2.9\% & 97.1\% & 16 & 2.8\% & 97.2\% & 172 & 2.8\% & 97.2\% & 656 \\
PageHinkley & 0.4\% & 99.6\% & 1 & 0.4\% & 99.6\% & 6 & 0.4\% & 99.6\% & 70 \\
DDM & 90.8\% & 9.2\% & 405 & 90.0\% & 9.2\% & 666 & 22.0\% & 9.2\% & 2050 \\
EDDM & 8.7\% & 91.3\% & 344 & 8.7\% & 91.3\% & 562 & 1.6\% & 91.3\% & 1759 \\
HDDM\_A & 94.0\% & 6.0\% & 60 & 94.0\% & 6.0\% & 168 & 52.5\% & 6.0\% & 2032 \\
HDDM\_W & 9.4\% & 90.5\% & 15 & 9.4\% & 90.5\% & 46 & 9.4\% & 90.5\% & 586 \\
\bottomrule
\end{tabular}
\end{table}

Table~\ref{tab:results} presents the full results; Figure~\ref{fig:detection_rates} visualizes TPR and FPR per method and scenario; Figure~\ref{fig:single_run} shows a representative single-run monitoring trace for the GSG scenario; Figure~\ref{fig:delay_dists} shows detection delay distributions across all detectors; and Figure~\ref{fig:cp_error} shows changepoint estimation error distributions for PITMonitor, which uniquely provides a changepoint estimate beyond the alarm time. Since null streams are drawn from the same pre-drift distribution across~all~the scenarios, each detector's FPR is near-identical across scenarios by construction.

\begin{figure}[b!]
\centering
\includegraphics[width=\linewidth]{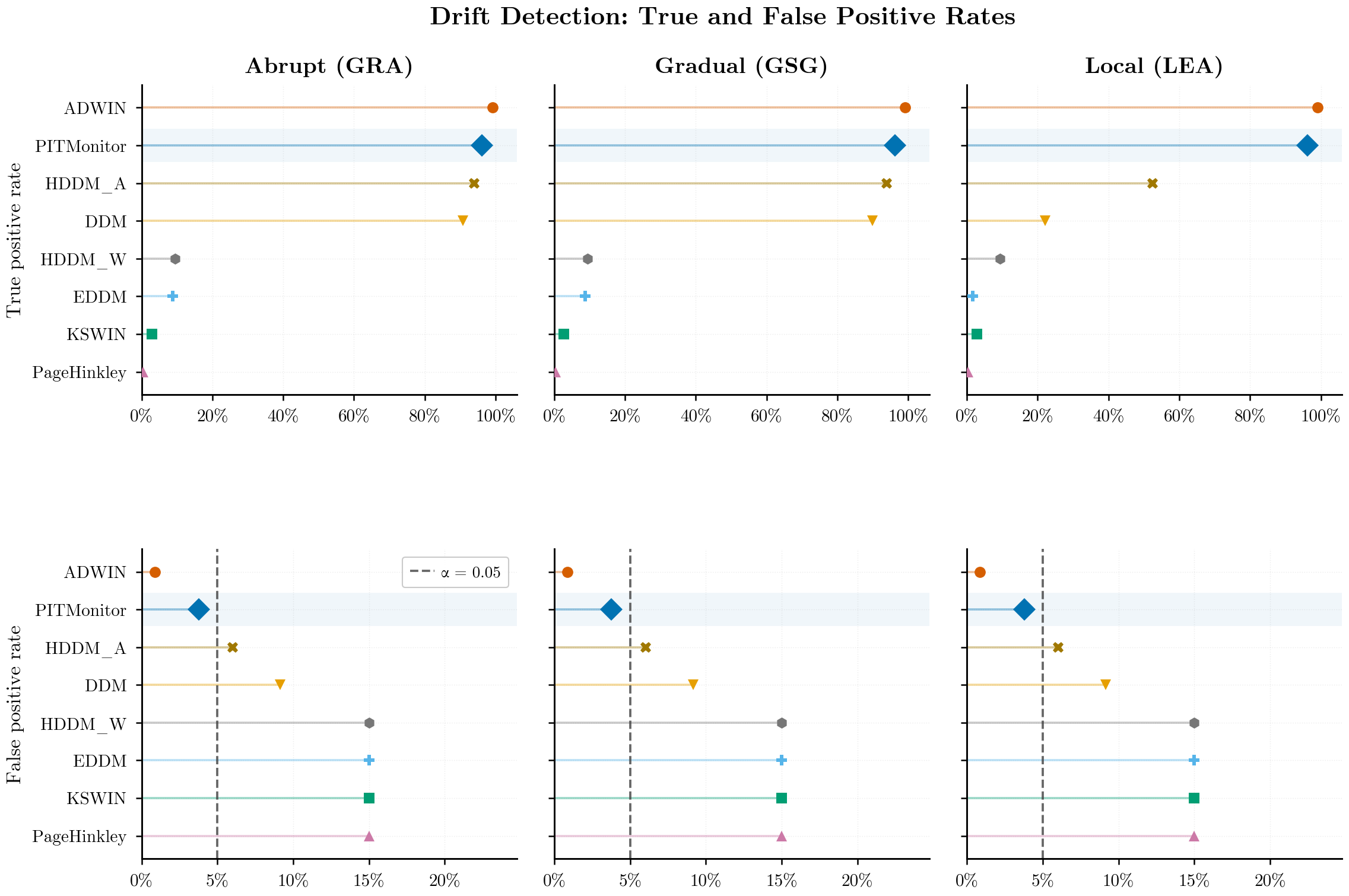}
\caption{TPR and FPR across all detectors and drift scenarios. The dashed line marks the nominal $\alpha = \expAlpha$.}
\label{fig:detection_rates}
\end{figure}

\begin{figure}[t!]
\centering
\includegraphics[width=\linewidth]{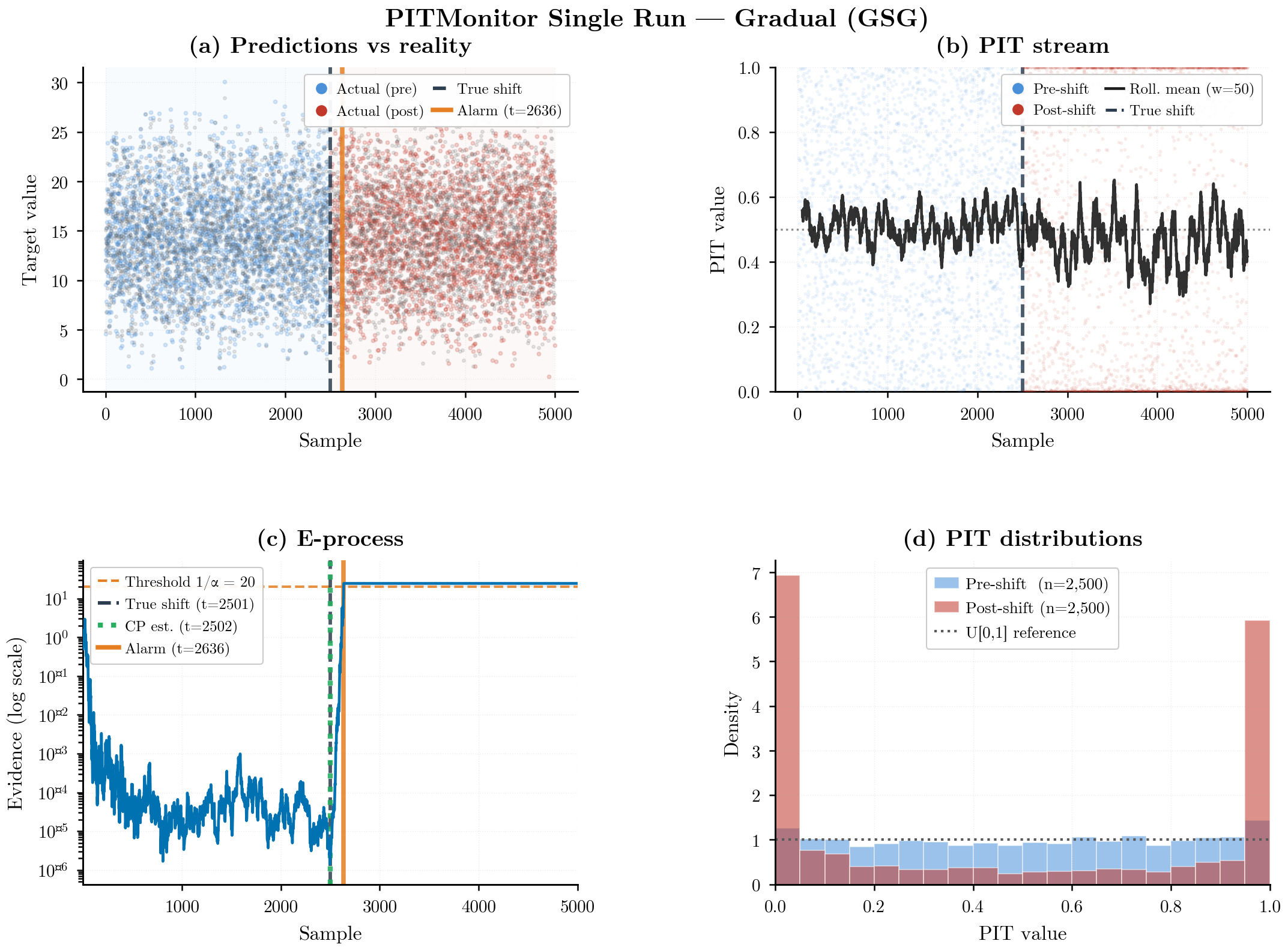}
\caption{Single-run PITMonitor trace. (a) Predicted vs.\ actual values. (b) PIT stream with a rolling mean. (c) Mixture e-process on a log scale, with a dashed line marking the threshold $1/\alpha$. (d) Pre-shift and post-shift PIT histograms.}
\label{fig:single_run}
\end{figure}

\paragraph{Type~I error control.}
PITMonitor achieves FPR of \resPMGRAFPR{} across all scenarios, consistent with the nominal $\alpha = \expAlpha$ guarantee in Theorem~\ref{thm:type1}. In this benchmark, ADWIN obtains a \resADGRAFPR{} empirical FPR on the \expNstable-sample null window, though this number is specific to the chosen horizon and dataset. Among the other baselines, DDM and HDDM\_A achieve \resDDMGRAFPR{} and \resHDDMAGRAFPR{} respectively, while KSWIN, PageHinkley, EDDM, and HDDM\_W all exceed $89\%$.

\paragraph{Global recurring abrupt drift (GRA).}
On GRA, PITMonitor achieves \resPMGRATPR{} TPR with a mean delay of \resPMGRADelay{} samples. ADWIN attains higher TPR (\resADGRATPR{}) and shorter delay (\resADGRADelay{}), though its FPR remains an empirical estimate tied to this finite monitoring window, while DDM and HDDM\_A reach a TPR of \resDDMGRATPR{} and \resHDDMAGRATPR{}, with a delay of \resDDMGRADelay{} and \resHDDMAGRADelay{} samples respectively.

\paragraph{Global slow gradual drift (GSG).}
On GSG, PITMonitor achieves \resPMGSGTPR{} TPR with a mean delay of \resPMGSGDelay{} samples. ADWIN again attains higher TPR (\resADGSGTPR{}) and shorter delay (\resADGSGDelay{}), while DDM and HDDM\_A reach \resDDMGSGTPR{} and \resHDDMAGSGTPR{} TPR, with a delay of \resDDMGSGDelay{} and \resHDDMAGSGDelay{} respectively.

\paragraph{Local expanding drift (LEA).}
On LEA, PITMonitor maintains \resPMLEATPR{} TPR with a mean delay of \resPMLEADelay{} samples. The increased delay is consistent with the expanding-drift structure: early phases induce weaker PIT distortion, so evidence accumulates slowly until later expansion phases push the e-process past the threshold. ADWIN detects earlier and with higher accuracy (\resADLEADelay{} sample delay, \resADLEATPR{} TPR), but no other baseline achieves a TPR above $52.5\%$.

\begin{figure}[t!]
\centering
\includegraphics[width=\linewidth]{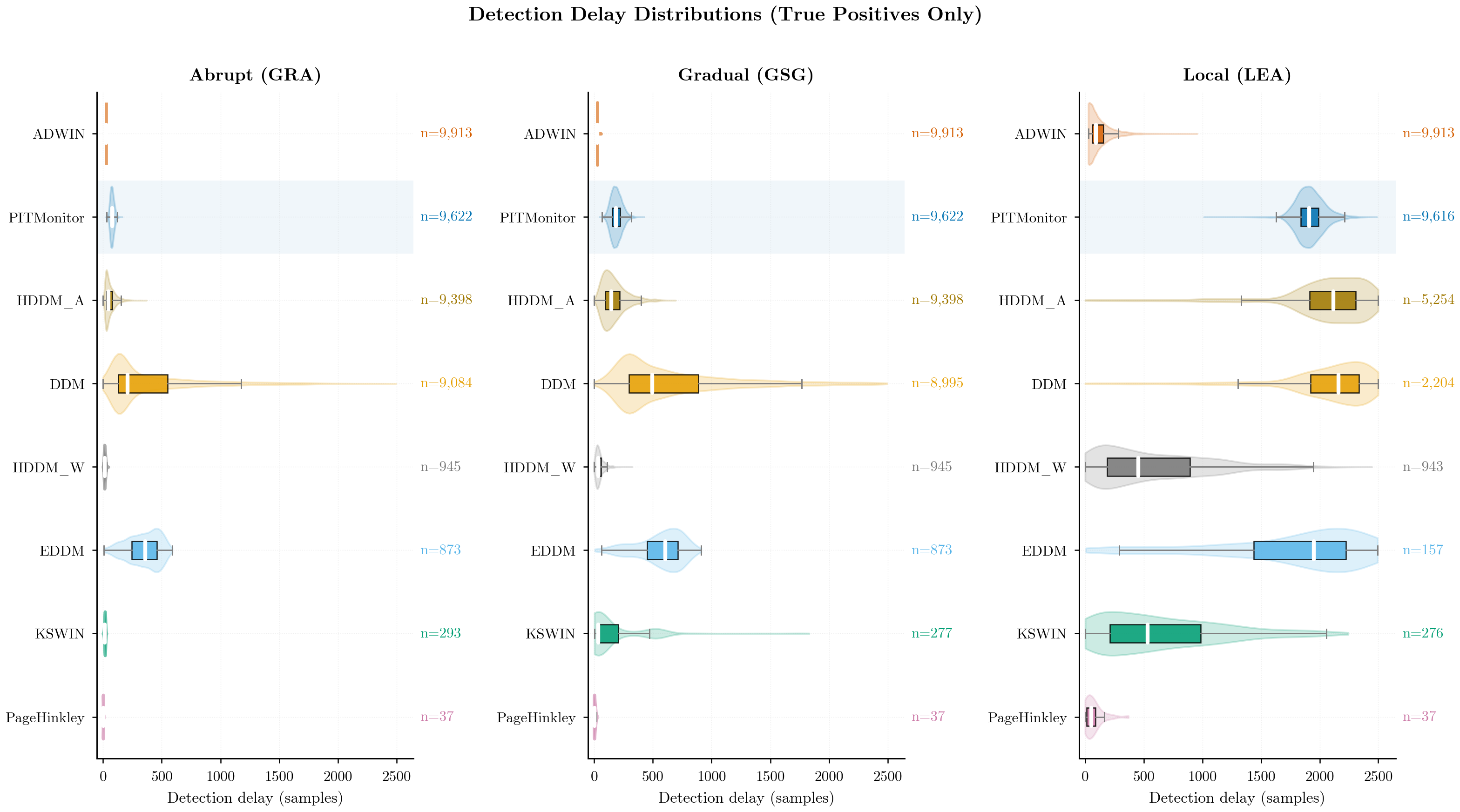}
\caption{Detection delay distributions. White bars mark the median; colored text indicates sample count.}
\label{fig:delay_dists}
\end{figure}

\paragraph{Detection delays.}
Figure~\ref{fig:delay_dists} reports delay distributions over true-positive trials only. Accordingly, it should be interpreted jointly with TPR; a detector with low TPR may have a tight delay distribution simply because it detects only the easiest cases. PITMonitor's delay distribution is tightest under GRA, where the shift is most immediate, and widest under LEA, where the shift is weakest at onset and evidence accumulates slowly. The other detectors show similar patterns, with ADWIN achieving the fastest detection among the detectors with a TPR > 90\%.

\begin{figure}[b!]
\centering
\includegraphics[width=\linewidth]{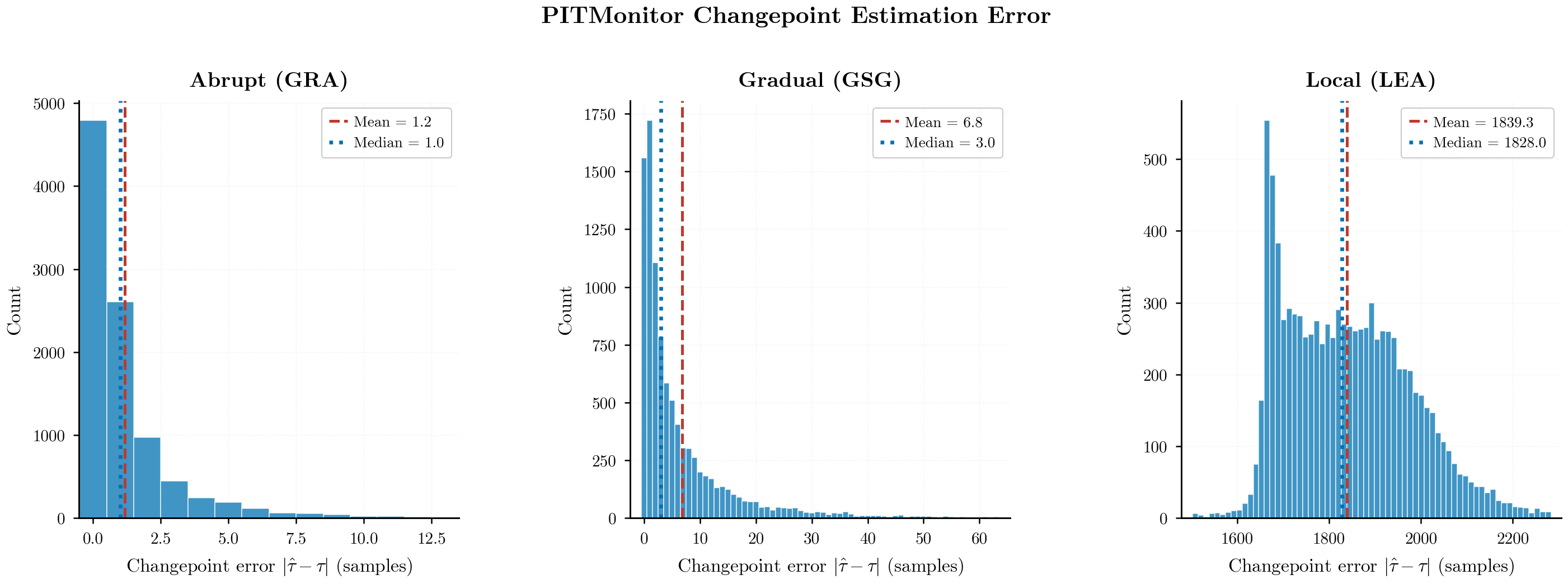}
\caption{Distribution of PITMonitor changepoint estimation error $|\hat\tau - \tau|$ across true-positive trials.}
\label{fig:cp_error}
\end{figure}

\paragraph{Changepoint localization.}
On GRA and GSG, mean absolute error is \resPMGRACPErr{} and \resPMGSGCPErr{} samples, indicating accurate localization under global drift. On LEA, mean error rises to \resPMLEACPErr{}, as the expanding multiple transition point structure violates the single-changepoint assumption. The changepoint estimator selects the split with strongest post-split non-uniformity, which typically falls in last phase rather than at initial onset. MAE here therefore reflects a mismatch to a single reference onset rather than a failure to identify the operationally relevant transition.

\section{Discussion}

\paragraph{Scope.}
PITMonitor targets continuous monitoring of deployed probabilistic models where false alarms are costly, monitoring is open-ended, and calibration drift specifically is operationally important. For one-time calibration checks or accuracy-only monitoring without false-alarm guarantees or changepoint localization, standard diagnostics remain appropriate.

\paragraph{Limitations.}\label{par:limitations}

\begin{enumerate}\setlength{\itemsep}{5.5pt}
 \item \textit{Expanding drift.} PITMonitor detects LEA drift with \resPMLEATPR{} TPR but \resPMLEADelay{} mean delay, reflecting slower evidence accumulation under weaker partial shifts. Changepoint localization is also less precise under multi-phase expansion since the estimator targets the strongest distributional change, which may lag the true onset (see Appendix~\ref{app:additional_checks}). Although LEA violates our single-changepoint assumption, PITMonitor identifies the dominant shift when multiple exist, which may be useful for operational monitoring. Multiple changepoint localization is an important direction for future work (Section~\ref{sec:conclusion}).

 \item \textit{Calibration improvements.} Because PITMonitor detects any shift in the PIT distribution, it does not distinguish between deterioration and improvement in calibration; a model that becomes better calibrated mid-stream may also trigger an alarm. The direction of drift can however be assessed by inspecting the PIT histogram following an alarm. Unlike the \texttt{river} baselines, which expose only a binary flag, PITMonitor's post-alarm histogram also provides a direct and interpretable link between the alarm and the underlying distributional change. Automatically distinguishing between improvements and deteriorations is left open and is an important direction for future work (Section~\ref{sec:conclusion}).

 \item \textit{Stationarity assumption.} PITMonitor's validity relies on PITs being i.i.d.\ under the null, so non-stationarity in the data stream may inflate the long-run false alarm rate. Our FriedmanDrift experiments did not reveal elevated empirical FPR, but this benchmark is stationary pre-drift by design; further evaluation on non-stationary streams is needed.
 
 \item \textit{Late changepoints after long null periods.} Proposition~\ref{prop:finite_time} shows that the warmup length scales with the pre-shift history length. A changepoint occurring after a long stable period therefore faces a deeply entrenched histogram and a decaying p-value signal (Section~\ref{sec:pvals}), potentially causing the e-process to grow too slowly to alarm before the signal erodes. In such settings, a smaller $B$ or periodic histogram reset may be effective.

 \item \textit{Persistence requirement.} PITMonitor reacts to accumulated evidence rather than single-step anomalies, so brief, rapidly reversing, or fast-moving deviations may yield delayed or absent alarms. PITMonitor is therefore best suited for persistent drift inducing a sustained shift in the PIT distribution. Pairing it with a faster residual-based detector can provide early warning of transient shifts, at the cost of weaker false-alarm guarantees.
\end{enumerate}

\paragraph{Parameter settings.}
The number of bins $B$ controls the bias-variance tradeoff in the density estimator: smaller $B$ yields more stable but slower-adapting estimates, while larger $B$ adapts faster at the cost of higher variance, and should be scaled with the expected number of monitoring samples. During the early monitoring phase (roughly the first $B$ observations) the histogram is dominated by pseudocounts, so e-values remain near 1 regardless of the data and sensitivity is low; practitioners expecting drift immediately after deployment may prefer a smaller $B$ or a warm-start period before monitoring begins. The significance level $\alpha$ should similarly be set against long-run false-alarm tolerance rather than treated as a fixed convention.

\section{Conclusion}
\label{sec:conclusion}

We presented PITMonitor, an anytime-valid method for monitoring the calibration of deployed probabilistic models. By detecting distributional shifts in probability integral transforms via a mixture e-process, it provides formal Type~I error control over unbounded monitoring as well as Bayesian changepoint estimation, properties that distinguish it from existing drift detectors.

On the three FriedmanDrift scenarios, PITMonitor attains competitive TPR on global drift while holding empirical FPR under the target level. On local expanding drift, detection delay is substantially longer and changepoint localization less accurate: the former reflects slower evidence accumulation under weaker partial shifts; the latter reflects the estimator targeting the dominant shift rather than the true onset when the single-changepoint assumption is violated.

Beyond numerical performance, PITMonitor's most distinctive property is its calibration-specificity; by operating on PITs rather than residuals or error rates, it is sensitive to miscalibration that leaves accuracy  unchanged, and the post-alarm histogram provides a direct and interpretable link between an alarm and the nature of the underlying distributional shift.

Several directions remain open. Improving power under partial shifts and extending reliable localization to multiple changepoints would address the two main limitations identified in the experiments. Multivariate outputs can initially be handled by reducing each prediction-outcome pair to a scalar PIT; a fuller treatment via Rosenblatt chains would handle joint miscalibration but requires access to the full conditional factorization. A natural next step is coupling detection with automatic post-alarm recalibration, for example by reweighting the predictive density by the empirical post-alarm PIT density as $\hat{f}(y) \propto p(y)\,\hat{q}(z)^{\lambda(e)}$, where $z = \hat{F}(y)$ and $\lambda(e)$ is a shrinkage weight learned from e-process evidence, tying correction strength directly to miscalibration severity. Distinguishing calibration deterioration from improvement would allow practitioners to act selectively on alarms, though the direction of drift is already partially recoverable from the post-alarm PIT histogram. Finally, a precise characterization of the behaviour of p-values under $H_1$ remains an open theoretical question.

\paragraph{Code Availability.} \url{https://github.com/tristan-farran/pitmon}.

\appendix
\section{Multi-changepoint localization check.}
\label{app:additional_checks}
In a staged stream in which drift intensity increases across three phases, PITMonitor localizes consistently to the highest-intensity phase rather than initial onset. MAE to the highest-intensity phase is \addLocMAEMax{} versus \addLocMAEOnset{} to onset, with \addLocCloserMax{} of \addLocN{} runs closer to the highest-intensity phase. This supports the interpretation that under multi-phase drift, the estimator tends to identify the dominant regime shift rather than the earliest detectable change.

\bibliographystyle{imsart-nameyear}
\bibliography{references}

\begin{thebibliography}{21}

\bibitem[\protect\citeauthoryear{Arnold, Henzi and Ziegel}{2023}]{arnold2023sequentially}
\begin{barticle}[author]
\bauthor{\bsnm{Arnold},~\bfnm{Sebastian}\binits{S.}}, \bauthor{\bsnm{Henzi},~\bfnm{Alexander}\binits{A.}} \AND \bauthor{\bsnm{Ziegel},~\bfnm{Johanna~F.}\binits{J.~F.}}
(\byear{2023}).
\btitle{Sequentially valid tests for forecast calibration}.
\bjournal{Annals of Applied Statistics}
\bvolume{17}.
\end{barticle}
\endbibitem

\bibitem[\protect\citeauthoryear{Bifet and Gavald{\`a}}{2007}]{bifet2007learning}
\begin{binproceedings}[author]
\bauthor{\bsnm{Bifet},~\bfnm{Albert}\binits{A.}} \AND \bauthor{\bsnm{Gavald{\`a}},~\bfnm{Ricard}\binits{R.}}
(\byear{2007}).
\btitle{Learning from time-changing data with adaptive windowing}.
In \bbooktitle{Proceedings of the 7th SIAM International Conference on Data Mining}.
\bpublisher{Society for Industrial and Applied Mathematics}.
\end{binproceedings}
\endbibitem

\bibitem[\protect\citeauthoryear{Dawid}{1984}]{dawid1984statistical}
\begin{barticle}[author]
\bauthor{\bsnm{Dawid},~\bfnm{A.~Philip}\binits{A.~P.}}
(\byear{1984}).
\btitle{Statistical theory: The prequential approach}.
\bjournal{Journal of the Royal Statistical Society: Series A}
\bvolume{147}.
\end{barticle}
\endbibitem

\bibitem[\protect\citeauthoryear{DeGroot and Fienberg}{1983}]{degroot1983comparison}
\begin{barticle}[author]
\bauthor{\bsnm{DeGroot},~\bfnm{Morris~H.}\binits{M.~H.}} \AND \bauthor{\bsnm{Fienberg},~\bfnm{Stephen~E.}\binits{S.~E.}}
(\byear{1983}).
\btitle{The comparison and evaluation of forecasters}.
\bjournal{Journal of the Royal Statistical Society: Series D}
\bvolume{32}.
\end{barticle}
\endbibitem

\bibitem[\protect\citeauthoryear{Diebold, Gunther and Tay}{1998}]{diebold1998evaluating}
\begin{barticle}[author]
\bauthor{\bsnm{Diebold},~\bfnm{Francis~X.}\binits{F.~X.}}, \bauthor{\bsnm{Gunther},~\bfnm{Todd~A.}\binits{T.~A.}} \AND \bauthor{\bsnm{Tay},~\bfnm{Anthony~S.}\binits{A.~S.}}
(\byear{1998}).
\btitle{Evaluating density forecasts with applications to financial risk management}.
\bjournal{International Economic Review}
\bvolume{39}.
\end{barticle}
\endbibitem

\bibitem[\protect\citeauthoryear{Fedorova et~al.}{2012}]{fedorova2012plugin}
\begin{binproceedings}[author]
\bauthor{\bsnm{Fedorova},~\bfnm{Valentina}\binits{V.}}, \bauthor{\bsnm{Gammerman},~\bfnm{Alex}\binits{A.}}, \bauthor{\bsnm{Nouretdinov},~\bfnm{Ilia}\binits{I.}} \AND \bauthor{\bsnm{Vovk},~\bfnm{Vladimir}\binits{V.}}
(\byear{2012}).
\btitle{Plug-in martingales for testing exchangeability on-line}.
In \bbooktitle{Proc. ICML}.
\end{binproceedings}
\endbibitem

\bibitem[\protect\citeauthoryear{Gneiting, Balabdaoui and Raftery}{2007}]{gneiting2007probabilistic}
\begin{barticle}[author]
\bauthor{\bsnm{Gneiting},~\bfnm{Tilmann}\binits{T.}}, \bauthor{\bsnm{Balabdaoui},~\bfnm{Fadoua}\binits{F.}} \AND \bauthor{\bsnm{Raftery},~\bfnm{Adrian~E.}\binits{A.~E.}}
(\byear{2007}).
\btitle{Probabilistic Forecasts, Calibration and Sharpness}.
\bjournal{Journal of the Royal Statistical Society Series B: Statistical Methodology}
\bvolume{69}.
\end{barticle}
\endbibitem

\bibitem[\protect\citeauthoryear{Gneiting and Katzfuss}{2014}]{gneiting2014probabilistic}
\begin{barticle}[author]
\bauthor{\bsnm{Gneiting},~\bfnm{Tilmann}\binits{T.}} \AND \bauthor{\bsnm{Katzfuss},~\bfnm{Matthias}\binits{M.}}
(\byear{2014}).
\btitle{Probabilistic forecasting}.
\bjournal{Annual Review of Statistics and Its Application}
\bvolume{1}.
\end{barticle}
\endbibitem

\bibitem[\protect\citeauthoryear{Gneiting and Raftery}{2007}]{gneiting2007strictly}
\begin{barticle}[author]
\bauthor{\bsnm{Gneiting},~\bfnm{Tilmann}\binits{T.}} \AND \bauthor{\bsnm{Raftery},~\bfnm{Adrian~E.}\binits{A.~E.}}
(\byear{2007}).
\btitle{Strictly proper scoring rules, prediction, and estimation}.
\bjournal{Journal of the American Statistical Association}
\bvolume{102}.
\end{barticle}
\endbibitem

\bibitem[\protect\citeauthoryear{Gr{\"u}nwald, de~Heide and Koolen}{2024}]{grunwald2024safe}
\begin{barticle}[author]
\bauthor{\bsnm{Gr{\"u}nwald},~\bfnm{Peter}\binits{P.}}, \bauthor{\bparticle{de} \bsnm{Heide},~\bfnm{Rianne}\binits{R.}} \AND \bauthor{\bsnm{Koolen},~\bfnm{Wouter~M.}\binits{W.~M.}}
(\byear{2024}).
\btitle{Safe testing}.
\bjournal{Journal of the Royal Statistical Society: Series B}
\bvolume{86}.
\end{barticle}
\endbibitem

\bibitem[\protect\citeauthoryear{Guo et~al.}{2017}]{guo2017calibration}
\begin{binproceedings}[author]
\bauthor{\bsnm{Guo},~\bfnm{Chuan}\binits{C.}}, \bauthor{\bsnm{Pleiss},~\bfnm{Geoff}\binits{G.}}, \bauthor{\bsnm{Sun},~\bfnm{Yu}\binits{Y.}} \AND \bauthor{\bsnm{Weinberger},~\bfnm{Kilian~Q.}\binits{K.~Q.}}
(\byear{2017}).
\btitle{On calibration of modern neural networks}.
In \bbooktitle{Proc. ICML}.
\end{binproceedings}
\endbibitem

\bibitem[\protect\citeauthoryear{Jeffreys}{1961}]{jeffreys1961}
\begin{bbook}[author]
\bauthor{\bsnm{Jeffreys},~\bfnm{Harold}\binits{H.}}
(\byear{1961}).
\btitle{Theory of Probability},
\bedition{3} ed.
\bpublisher{Oxford University Press}.
\end{bbook}
\endbibitem

\bibitem[\protect\citeauthoryear{Jenks}{2014}]{sortedcontainers}
\begin{bmisc}[author]
\bauthor{\bsnm{Jenks},~\bfnm{Grant}\binits{G.}}
(\byear{2014}).
\btitle{{sortedcontainers}: Sorted Containers}.
\end{bmisc}
\endbibitem

\bibitem[\protect\citeauthoryear{Lipton, Wang and Smola}{2018}]{lipton2018detecting}
\begin{binproceedings}[author]
\bauthor{\bsnm{Lipton},~\bfnm{Zachary}\binits{Z.}}, \bauthor{\bsnm{Wang},~\bfnm{Yu-Xiang}\binits{Y.-X.}} \AND \bauthor{\bsnm{Smola},~\bfnm{Alex}\binits{A.}}
(\byear{2018}).
\btitle{Detecting and correcting for label shift with black box predictors}.
In \bbooktitle{Proc. ICML}.
\end{binproceedings}
\endbibitem

\bibitem[\protect\citeauthoryear{Montiel and Others}{2021}]{montiel2021river}
\begin{barticle}[author]
\bauthor{\bsnm{Montiel},~\bfnm{Jacob}\binits{J.}} \AND \bauthor{\bsnm{Others}}
(\byear{2021}).
\btitle{River: Machine learning for streaming data in {P}ython}.
\bjournal{Journal of Machine Learning Research}
\bvolume{22}.
\end{barticle}
\endbibitem

\bibitem[\protect\citeauthoryear{Rabanser, G{\"u}nnemann and Lipton}{2019}]{rabanser2019failing}
\begin{binproceedings}[author]
\bauthor{\bsnm{Rabanser},~\bfnm{Stephan}\binits{S.}}, \bauthor{\bsnm{G{\"u}nnemann},~\bfnm{Stephan}\binits{S.}} \AND \bauthor{\bsnm{Lipton},~\bfnm{Zachary}\binits{Z.}}
(\byear{2019}).
\btitle{Failing loudly: An empirical study of methods for detecting dataset shift}.
In \bbooktitle{Advances in Neural Information Processing Systems (NeurIPS)}
\bvolume{32}.
\end{binproceedings}
\endbibitem

\bibitem[\protect\citeauthoryear{Ramdas et~al.}{2023}]{ramdas2023game}
\begin{barticle}[author]
\bauthor{\bsnm{Ramdas},~\bfnm{Aaditya}\binits{A.}}, \bauthor{\bsnm{Gr{\"u}nwald},~\bfnm{Peter}\binits{P.}}, \bauthor{\bsnm{Vovk},~\bfnm{Vladimir}\binits{V.}} \AND \bauthor{\bsnm{Shafer},~\bfnm{Glenn}\binits{G.}}
(\byear{2023}).
\btitle{Game-theoretic statistics and safe anytime-valid inference}.
\bjournal{Statistical Science}
\bvolume{38}.
\end{barticle}
\endbibitem

\bibitem[\protect\citeauthoryear{Shin, Ramdas and Rinaldo}{2024}]{shin2022detectors}
\begin{barticle}[author]
\bauthor{\bsnm{Shin},~\bfnm{Jaehyeok}\binits{J.}}, \bauthor{\bsnm{Ramdas},~\bfnm{Aaditya}\binits{A.}} \AND \bauthor{\bsnm{Rinaldo},~\bfnm{Alessandro}\binits{A.}}
(\byear{2024}).
\btitle{E-detectors: A nonparametric framework for sequential change detection}.
\bjournal{The New England Journal of Statistics in Data Science}
\bvolume{2}.
\end{barticle}
\endbibitem

\bibitem[\protect\citeauthoryear{Ville}{1939}]{ville1939}
\begin{bbook}[author]
\bauthor{\bsnm{Ville},~\bfnm{Jean}\binits{J.}}
(\byear{1939}).
\btitle{{\'E}tude Critique de la Notion de Collectif}.
\bpublisher{Gauthier-Villars}, \baddress{Paris}.
\end{bbook}
\endbibitem

\bibitem[\protect\citeauthoryear{Vovk, Gammerman and Shafer}{2005}]{vovk2005algorithmic}
\begin{bbook}[author]
\bauthor{\bsnm{Vovk},~\bfnm{Vladimir}\binits{V.}}, \bauthor{\bsnm{Gammerman},~\bfnm{Alex}\binits{A.}} \AND \bauthor{\bsnm{Shafer},~\bfnm{Glenn}\binits{G.}}
(\byear{2005}).
\btitle{Algorithmic Learning in a Random World}.
\bpublisher{Springer}.
\end{bbook}
\endbibitem

\bibitem[\protect\citeauthoryear{Vovk and Wang}{2021}]{vovk2021values}
\begin{barticle}[author]
\bauthor{\bsnm{Vovk},~\bfnm{Vladimir}\binits{V.}} \AND \bauthor{\bsnm{Wang},~\bfnm{Ruodu}\binits{R.}}
(\byear{2021}).
\btitle{E-values: Calibration, combination, and applications}.
\bjournal{Annals of Statistics}
\bvolume{49}.
\end{barticle}
\endbibitem

\end{thebibliography}

\end{document}